\documentclass{article}
\usepackage{amsmath,amsfonts,amssymb,amsthm,stmaryrd}
\usepackage{hyperref,fullpage}
\usepackage[boxed,vlined]{algorithm2e}

\newtheorem{lemma}{Lemma}

\newtheorem{theorem}{Theorem}
\newtheorem*{fact}{Fact}

\providecommand{\NN}{\mathbb{N}}
\providecommand{\RR}{\mathbb{R}}
\providecommand{\RRnn}{\RR_{\geq 0}}
\providecommand{\UU}{\mathcal{U}}
\providecommand{\otherU}{\mathcal{V}}
\providecommand{\otherUset}[1]{V_{#1}}
\providecommand{\II}{\mathcal{I}}
\providecommand{\mm}{\mathcal{M}}
\providecommand{\avgX}{\overline{X}}
\providecommand{\gapp}{\tilde{g}}

\providecommand{\ib}[1]{\llbracket #1 \rrbracket}
\providecommand{\elem}[1]{\mathfrak{#1}}
\providecommand{\numelem}[2]{h(#1 ; #2)}
\providecommand{\decor}[1]{\breve{#1}}
\providecommand{\ind}[1]{1_{#1}}
\providecommand{\sinit}{S_{\mathrm{init}}}

\providecommand{\mcoeff}[2]{m_{#1,#2}}
\providecommand{\Beta}{\mathrm{B}}
\DeclareMathOperator*{\EE}{\mathbb{E}}
\DeclareMathOperator*{\argmax}{argmax}

\DeclareMathOperator{\Ber}{Ber}

\SetKw{Break}{break}
\SetKwComment{Comment}{\{}{\}}
\SetAlCapSkip{1ex}

\title{Monotone Submodular Maximization over a Matroid \\ via Non-Oblivious Local Search}
\author{Yuval Filmus\thanks{\href{mailto:yuvalf@berkeley.edu}{yuvalf@berkeley.edu}} and Justin Ward\thanks{\href{mailto:J.D.Ward@dcs.warwick.ac.uk}{J.D.Ward@dcs.warwick.ac.uk} Work supported by EPSRC grant EP/J021814/1.
}}

\begin{document}
\maketitle

\begin{abstract}
We present an optimal, combinatorial $1-1/e$ approximation algorithm for monotone submodular optimization over a matroid constraint. Compared to the continuous greedy algorithm (Calinescu, Chekuri, P\'al and Vondr\'ak, 2008), our algorithm is extremely simple and requires no rounding. It consists of the greedy algorithm followed by local search. Both phases are run not on the actual objective function, but on a related auxiliary potential function, which is also monotone submodular. 

In our previous work on maximum coverage (Filmus and Ward, 2012), the potential function gives more weight to elements covered multiple times. We generalize this approach from coverage functions to arbitrary monotone submodular functions. When the objective function is a coverage function, both definitions of the potential function coincide.

Our approach generalizes to the case where the monotone submodular function has restricted curvature. For any curvature $c$, we adapt our algorithm to produce a $(1-e^{-c})/c$ approximation. This matches results of Vondr\'ak (2008), who has shown that the continuous greedy algorithm produces a $(1-e^{-c})/c$ approximation when the objective function has curvature $c$, and proved that achieving any better approximation ratio is impossible in the value oracle model.
\end{abstract}

\section{Introduction} \label{sec:introduction}
In this paper, we consider the problem of maximizing a monotone submodular function $f$, subject to a single matroid constraint.  Formally, let $\UU$ be a set of $n$ elements and let $f\colon 2^\UU \to \RR$ be a function assigning a value to each subset of $\UU$.  We say that $f$ is \emph{submodular} if 
\[
f(A) + f(B) \ge f(A \cup B) + f(A \cap B)
\]
for all $A, B \subseteq \UU$.  If additionally, $f(A) \le f(B)$ whenever $A \subseteq B$, we say that $f$ is \emph{monotone submodular}.  Submodular functions exhibit (and are, in fact, alternately characterized by) the property of diminishing returns---if $f$ is submodular then $f(A \cup \{x\}) - f(A) \le f(B \cup \{x\}) - f(B)$ for all $B \subseteq A$.  Hence, they are useful for modeling economic and game-theoretic scenarios, as well as various combinatorial problems. In a general monotone submodular maximization problem, we are given a value oracle for $f$ and a membership oracle for some distinguished collection $\II \subseteq 2^\UU$ of \emph{feasible sets}, and our goal is to find a member of $\II$ that maximizes the value of $f$. We assume further that $f$ is \emph{normalized} so that $f(\emptyset) = 0$.

We consider the restricted setting in which the collection $\II$ forms a matroid.
Matroids are intimately connected to combinatorial optimization: the problem of optimizing a linear function over a hereditary set system (a set system closed under taking subsets) is solved optimally for all possible functions by the standard greedy algorithm if and only if the set system is a matroid~\cite{Rado-1957,Edmonds-1971}.  

In the case of a \emph{monotone submodular} objective function, the standard greedy algorithm, which takes at each step the element yielding the largest increase in $f$ while maintaining independence, is (only) a $1/2$-approximation \cite{Fisher-1978}.  Recently, Calinescu et al.~\cite{Calinescu-2007,Vondrak-2008,Calinescu-2011} have developed a $(1 - 1/e)$-approximation for this problem via the \emph{continuous greedy algorithm}, which is essentially a steepest ascent algorithm running in continuous time (when implemented, a suitably discretized version is used), producing a fractional solution. The fractional solution is rounded using pipage rounding \cite{Ageev-2004} or swap rounding \cite{Chekuri-2010}.

Feige \cite{Feige-1998} has shown that improving the bound $(1 - 1/e)$ is NP-hard.   Nemhauser and Wolsey \cite{Nemhauser-1978a} have shown that any improvement over $(1 - 1/e)$  requires an exponential number of queries in the value oracle setting.

Following Vondr\'ak \cite{Vondrak-2008a}, we also consider the case when $f$ has restricted curvature. We say that $f$ has \emph{curvature $c$} if for any two disjoint $A,B \subseteq \UU$,
\[
 f(A \cup B) \geq f(A) + (1-c)f(B).
\]
When $c = 1$, this is a restatement of monotonicity of $f$, and when $c=0$, linearity of $f$. Vondr\'ak~\cite{Vondrak-2008a} has shown that the continuous greedy algorithm produces a $(1-e^{-c})/c$ approximation when $f$ has curvature~$c$. Furthermore, he has shown that any improvement over $(1-e^{-c})/c$ requires an exponential number of queries in the value oracle setting. 


\subsection{Our contribution} \label{sec:contribution}
In this paper, we propose a conceptually simple randomized polynomial time local search algorithm for the problem of monotone submodular matroid maximization.  Like the continuous greedy algorithm, our algorithm delivers the optimal $(1 - 1/e)$-approximation.  However, unlike the continuous greedy algorithm, our algorithm is entirely combinatorial, in the sense that it deals only with integral solutions to the problem and hence involves no rounding procedure.  As such, we believe that the algorithm may serve as a gateway to further improved algorithms in contexts where pipage rounding and swap rounding break down, such as submodular maximization subject to multiple matroid constraints.  

Our main results are a combinatorial $1 - 1/e - \epsilon$ approximation algorithm for monotone submodular matroid maximization, running in randomized time $\tilde{O}(\epsilon^{-3}r^4n)$, and a combinatorial $1 - 1/e$ approximation algorithm running in randomized time $\tilde{O}(r^7n^2)$, where $r$ is the rank of the given matroid and $n$ is the size of its ground set.  
Our algorithm further generalizes to the case in which the submodular function has curvature $c$\footnote{In fact, it is enough to assume that $f(A \cup B) \geq f(A) + (1-c) f(B)$ for any two disjoint \emph{independent sets} $A,B$.}. In this case the approximation ratios obtained are $(1-e^{-c})/c - \epsilon$ and $(1-e^{-c})/c$, respectively, again matching the performance of the continuous greedy algorithm \cite{Vondrak-2008a}.  Unlike the continuous greedy algorithm, our algorithm requires knowledge of $c$. However, by enumerating over values of $c$ we are able to obtain a combinatorial $(1-e^{-c})/c$ algorithm even in the case that $f$'s curvature is unknown.\footnote{For technical reasons, we require that $f$ has curvature bounded away from zero in this case.}

Our algorithmic approach is based on local search.  In classical local search, the algorithm starts at an arbitrary solution, and proceeds by iteratively making small changes that improve the objective function, until no such improvement can be made.  A natural, worst-case guarantee on the approximation performance of a local search algorithm is the \emph{locality ratio}, given as $\min f(S)/f(O)$, where $S$ is a locally optimal solution (i.e. a solution which cannot be improved by the small changes considered by the algorithm), $O$ is a global optimum, and $f$ is the objective function.  

In many cases, classical local search may have a very poor locality ratio, implying that a locally-optimal solution may be of significantly lower quality than the global optimum. For example, for monotone submodular maximization over a matroid, the locality ratio for an algorithm changing a single element at each step is $1/2$ \cite{Fisher-1978}.  \emph{Non-oblivious} local search, a technique first proposed by Alimonti~\cite{Alimonti-1994} and by Khanna, Motwani, Sudan and Vazirani~\cite{Khanna-1999},
attempts to avoid this problem by making use of a secondary potential function to guide the search.  By carefully choosing this auxiliary function, we ensure that poor local optima with respect to the original objective function are no longer local optima with respect to the new potential function.  This is the approach that we adopt in the design of our local search algorithm.  Specifically, we consider a simple local search algorithm in which the value of a solution is measured with respect to a carefully designed potential function $g$, rather than the submodular objective function $f$.  We show that solutions which are locally optimal with respect to $g$ have significantly higher worst-case quality (as measured by the problem's original potential function $f$) than those which are locally optimal with respect to $f$.

In previous work~\cite{Filmus-2012}, we designed an optimal non-oblivious local search algorithm for the restricted case of maximum coverage subject to a matroid constraint.  In this problem, we are given a weighted universe of elements, a collection of sets, and a matroid defined on this collection.  The goal is to find a collection of sets that is independent in the matroid and covers elements of maximum total weight.  The non-oblivious potential function used in~\cite{Filmus-2012} gives extra weight to solutions that cover elements multiple times.  That is, the potential function depends critically on the coverage representation of the objective function.  In the present work, we extend this approach to \emph{general} monotone submodular functions. This presents two challenges: defining a non-oblivious potential function without reference to the coverage representation, and analyzing the resulting algorithm.  

In order to define the general potential function, we construct a generalized variant of the potential function from~\cite{Filmus-2012} that does not require a coverage representation.  Instead, the potential function aggregates information obtained by applying the objective function to all subsets of the input, weighted according to their size.  Intuitively, the resulting potential function gives extra weight to solutions that contain a large number of good sub-solutions, or equivalently, remain good solutions in expectation when elements are removed by a random process.  An appropriate setting of the weights defining our potential function yields a function which coincides with the previous definition for coverage functions, but still makes sense for arbitrary monotone submodular functions.

The analysis of the algorithm in~\cite{Filmus-2012} is relatively straightforward. For each type of element in the universe of the coverage problem, we must prove a certain inequality among the coefficients defining the potential function. In the general setting, however, we need to construct a proof using only the inequalities given by monotonicity and submodularity.  The resulting proof is non-obvious and delicate.

\medskip

This paper extends and simplifies previous work by the same authors. The paper~\cite{Filmus-2012b}, appearing in FOCS~2012, only discusses the case $c=1$.  The general case is discussed in~\cite{Filmus-2012c}, which appears in ArXiv.  
The potential functions used to guide the non-oblivious local search in both the unrestricted curvature case \cite{Filmus-2012b} and the maximum coverage case \cite{Filmus-2012} are special cases of the function $g$ we discuss in the present paper.\footnote{The functions from~\cite{Filmus-2012b,Filmus-2012c} are defined in terms of certain coefficients $\gamma$, which depend on a parameter $E$. Our definition corresponds to the choice $E = e^c$.  We examine the case of coverage functions in more detail in Section~\ref{sec:maximum-coverage}.}
An exposition of the ideas of both~\cite{Filmus-2012} and~\cite{Filmus-2012c} can be found in the second author's thesis~\cite{Ward-2012a}. In particular, the thesis explains how the auxiliary objective function can be determined by solving a linear program, both in the special case of maximum coverage and in the general case of monotone submodular functions with restricted curvature.

\subsection{Related work} \label{sec:related-work}
Fisher, Nemhauser and Wolsey \cite{Nemhauser-1978,Fisher-1978} analyze greedy and local search algorithms for submodular maximization subject to various constraints, including single and multiple matroid constraints.  They obtain some of the earliest results in the area, including a $1/(k + 1)$-approximation algorithm for monotone submodular maximization subject to $k$ matroid constraints.  A recent survey by Goundan and Schulz \cite{Goundan-2007} reviews many results pertaining to the greedy algorithm for submodular maximization. 

More recently, Lee, Sviridenko and Vondr\'{a}k~\cite{Lee-2010a} consider the problem of both monotone and non-monotone submodular maximization subject to multiple matroid constraints, attaining a $1/(k + \epsilon)$-approx\-i\-mation for monotone submodular maximization subject to $k \ge 2$ constraints using local search.  Feldman et al.~\cite{Feldman-2011} show that a local search algorithm attains the same bound for the related class of $k$-exchange systems, which includes the intersection of $k$ strongly base orderable matroids, as well as the independent set problem in $(k + 1)$-claw free graphs.  Further work by Ward \cite{Ward-2012} shows that a non-oblivious local search routine attains an improved approximation ratio of $2/(k + 3) - \epsilon$ for this class of problems.

In the case of unconstrained non-monotone maximization, Feige, Mirrokni and Vondr\'{a}k \cite{Feige-2007} give a $2/5$-approximation algorithm via a randomized local search algorithm, and give an upper bound of $1/2$ in the value oracle model.  Gharan and Vondr\'{a}k \cite{Gharan-2011} improved the algorithmic result to $0.41$ by enhancing the local search algorithm with ideas borrowed from simulated annealing.  Feldman, Naor and Schwarz \cite{Feldman-2011a} later improved this to $0.42$ by using a variant of the continuous greedy algorithm. Buchbinder, Feldman, Naor and Schwartz have recently obtained an optimal $1/2$-approximation algorithm \cite{Feldman-2012}.

In the setting of constrained non-monotone submodular maximization, Lee et al.~\cite{Lee-2009} give a $1/(k+2+\frac{1}{k}+\epsilon)$-approximation algorithm for the case of $k$ matroid constraints and a $(1/5 - \epsilon)$-approximation algorithm for $k$ knapsack constraints.  Further work by Lee, Sviridenko and Vondr\'{a}k \cite{Lee-2010a} improves the approximation ratio in the case of $k$ matroid constraints to $1/(k + 1 + \frac{1}{k-1} + \epsilon)$.  Feldman et al.~\cite{Feldman-2011} attain this ratio for $k$-exchange systems.  In the case of non-monotone submodular maximization subject to a \emph{single} matroid constraint, Feldman, Naor and Schwarz \cite{Feldman-2011b} show that a version of the continuous greedy algorithm attains an approximation ratio of $1/e$. They additionally unify various applications of the continuous greedy algorithm and obtain improved approximations for non-monotone submodular maximization subject to a matroid constraint or $O(1)$ knapsack constraints.

\subsection{Organization of the paper} \label{sec:organization}

We begin by giving some basic definitions in Section \ref{sec:definitions}.  In Section \ref{sec:algorithm} we introduce our basic, non-oblivious local search algorithm, which makes use of an auxiliary potential function $g$.  In Section \ref{sec:g}, we give the formal definition of $g$, together with several of its properties. Unfortunately, exact computation of the function $g$ requires evaluating $f$ on an exponential number of sets.  In Section \ref{sec:algorithm-1}  we present a simplified analysis of our algorithm, under the assumption that an oracle for computing the function $g$ is given.  In Section \ref{sec:rp} we then show how to remove this assumption to obtain our main, randomized polynomial time algorithm.  The resulting algorithm uses a polynomial-time random sampling procedure to compute the function $g$ \emph{approximately}.  Finally, some simple extensions of our algorithm are described in Section \ref{sec:extensions}.

\section{Definitions} \label{sec:definitions}
\paragraph{Notation} If $B$ is some Boolean condition, then
\[ \ib{B} = \begin{cases} 1 & \text{if $B$ is true}, \\ 0 & \text{if $B$ is false.} \end{cases} \]
For $n$ a natural number, $[n] = \{1,\ldots,n\}$. We use $H_k$ to denote the $k$th Harmonic number,
\[ H_k = \sum_{t=1}^k \frac{1}{t}. \]
It is well-known that $H_k = \Theta(\ln k)$, where $\ln k$ is the natural logarithm.

For $S$ a set and $x$ an element, we use the shorthands $S + x = S \cup \{x\}$ and $S - x = S \setminus \{x\}$. We use the notation $S + x$ even when $x \in S$, in which case $S + x = S$, and the notation $S - x$ even when $x \notin S$, in which case $S - x = S$.

Let $\UU$ be a set. A \emph{set-function} $f$ on $\UU$ is a function $f\colon 2^{\UU} \rightarrow \RR$ whose arguments are subsets of $\UU$. For $x \in \UU$, we use $f(x) = f(\{x\})$. For $A,B \subseteq \UU$, the \emph{marginal} of $B$ with respect to $A$ is
\[ f_A(B) = f(A \cup B) - f(A). \]

\paragraph{Properties of set-functions} A set-function $f$ is \emph{normalized} if $f(\emptyset) = 0$. It is \emph{monotone} if whenever $A \subseteq B$ then $f(A) \leq f(B)$. It is \emph{submodular} if whenever $A \subseteq B$ and $C$ is disjoint from $B$, $f_A(C) \geq f_B(C)$. If $f$ is monotone, we need not assume that $B$ and $C$ are disjoint. Submodularity is equivalently characterized by the inequality \[ f(A) + f(B) \geq f(A\cup B) + f(A\cap B), \]
for all $A$ and $B$.

The set-function $f$ has \emph{curvature} $c$ if for all $A \subseteq \UU$ and $x \notin A$, $f_A(x) \geq (1-c) f(x)$. Equivalently, $f_A(B) \geq (1-c) f(B)$ for all disjoint $A,B \subseteq \UU$. Note that if $f$ has curvature $c$ and $c' \geq c$, then $f$ also has curvature $c'$. Every normalized monotone function thus has curvature $1$. A normalized function with curvature $0$ is linear; that is, $f_A(x) = f(x)$.


\paragraph{Matroids} A matroid $\mm=(\UU,\II)$ is composed of a ground set $\UU$ and a non-empty collection $\II$ of subsets of $\UU$ satisfying the following two properties: (1) If $A \in \II$ and $B \subseteq A$ then $B \in \II$; (2) If $A,B \in \II$ and $|A| > |B|$ then $B + x \in \II$ for some $x \in A \setminus B$.

The sets in $\II$ are called \emph{independent sets}. Maximal independent sets are known as \emph{bases}. Condition (2) implies that all bases of the matroid have the same size.  This common size is called the \emph{rank} of the matroid.

One simple example is a \emph{partition matroid}. The universe $\UU$ is partitioned into $r$ parts $\UU_1,\ldots,\UU_r$, and a set is independent if it contains at most one element from each part.

If $A$ is an independent set, then the \emph{contracted matroid} $\mm/A=(\UU\setminus A,\II/A)$ is given by
\[ \II/A = \{ B \subseteq \UU \setminus A : A \cup B \in \mm \}. \]

\paragraph{Monotone submodular maximization} An instance of monotone submodular maximization is given by $(\mm=(\UU,\II),f)$, where $\mm$ is a matroid and $f$ is a set-function on $\UU$ which is normalized, monotone and submodular.

The \emph{optimum} of the instance is
\[ f^* = \max_{O \in \II} f(O). \]
Because $f$ is monotone, the maximum is always attained at some basis.

We say that a set $S \in \II$ is an \emph{$\alpha$-approximate} solution if $f(S) \geq \alpha f(O)$. Thus $0 \leq \alpha \leq 1$.  We say that an algorithm has an \emph{approximation ratio} of $\alpha$ (or, simply that an algorithm \emph{provides an $\alpha$-approximation}) if it produces an $\alpha$-approximate solution on every instance.

\section{The algorithm} \label{sec:algorithm}
Our non-oblivious local search algorithm is shown in Algorithm~\ref{alg:1}.  The algorithm takes the following input parameters:
\begin{itemize}
\item A matroid $\mm = (\UU,\II)$, given as a ground set $\UU$ and a membership oracle for some collection $\II \subseteq 2^\UU$ of independent sets, which returns whether or not $X \in \II$ for any $X \subseteq \UU$.
\item A monotone submodular function $f\colon 2^\UU \to \RR_{\ge 0}$, given as a value oracle that returns $f(X)$ for any $X \subseteq \UU$.
\item An upper bound $c \in (0,1]$ on the curvature of $f$.  The case in which the curvature of $f$ is unrestricted corresponds to $c = 1$.
\item A convergence parameter $\epsilon$.
\end{itemize}
Throughout the paper, we let $r$ denote the rank of $\mm$ and $n = |\UU|$.
\begin{algorithm}
\label{alg:1}
\KwIn{$\mm = (\UU,\II),\ f,\ c,\ \epsilon$}
Set $\epsilon_1 = \frac{\epsilon}{rH_r}$\;
Let $\sinit$ be the result of running the 
standard greedy algorithm on $(\mm, g)$\;
$S \gets \sinit$\;
\Repeat{No exchange is made}
  {\ForEach{element $e \in S$ and $x \in \UU \setminus S$}{
    $S' \gets S - e + x$\;
    \If{$S' \in \II$ and $g(S') > (1 + \epsilon_1)g(S)$}{
      $S \gets S'$\;
      \Break\;
    }
  }
}
\Return $S$\;
\caption{The non-oblivious local search algorithm}
\end{algorithm}

The algorithm starts from an initial greedy solution $\sinit$, and proceeds by repeatedly exchanging one element $e$ in the current solution $S$ for one element $x$ not in $S$, with the aim of obtaining an improved independent set $S' \in \II$.  In both the initial greedy phase and the following local search phase, the quality of the solution is measured not with respect to $f$, but rather with respect to an auxiliary potential function $g$ (as we discuss shortly, we in fact must use an estimate $\gapp$ for $g$), which is determined by the rank of $\mm$ and the value of the curvature bound $c$.

We give a full definition of $g$ in Section \ref{sec:g}.  The function is determined by a sequence of coefficients depending on the upper bound $c$ on the curvature of $f$.  Evaluating the function $g$ exactly will require an exponential number of value queries to $f$.  Nonetheless, in Section \ref{sec:rp} we show how to modify Algorithm \ref{alg:1} by using a random sampling procedure to approximate $g$.  The resulting algorithm has the desired approximation guarantee with high probability and runs in polynomial time.

At each step we require that an improvement increase $g$ by a factor of at least $1 + \epsilon_1$.  This, together with the initial greedy choice of $\sinit$, ensures that Algorithm \ref{alg:1} converges in time polynomial in $r$ and $n$, at the cost of a slight loss in its locality gap.  In Section \ref{sec:extensions} we describe how the small resulting loss in the approximation ratio can be recovered, both in the case of Algorithm \ref{alg:1}, and in the randomized, polynomial-time variant we consider in Section \ref{sec:rp}.

\section{The auxiliary objective function \texorpdfstring{$g$}{g}} \label{sec:g}
We turn to the remaining task needed for completing the definition of Algorithm \ref{alg:1}: giving a definition of the potential function $g$.
The construction we use for $g$ will necessarily depend on $c$, but because we have fixed an instance, we shall omit this dependence from our notation, in order to avoid clutter.

\subsection{Definition of \texorpdfstring{$g$}{g}} \label{sec:g-definition}
We now present a definition of our auxiliary potential function $g$.  Our goal is to give extra value to solutions $S$ that are robust with respect to small changes.  That is, we would like our potential function to assign higher value to solutions that retain their quality even when some of their elements are removed by future iterations of the local search algorithm.  We model this general notion of robustness by considering a random process that obtains a new solution $T$ from the current solution $S$ by independently discarding each element of $S$ with some probability.  Then we use the expected value of $f(T)$ to define our potential function $g$

It will be somewhat more intuitive to begin by relating the \emph{marginals} $g_A$ of $g$ to the \emph{marginals} $f_A$ of $f$, rather than directly defining the values of $g$ and $f$.  We begin by considering some simple properties that we would like to hold for the marginals, and eventually give a concrete definition of $g$, showing that it has these properties.

Let $A$ be some subset of $\UU$ and consider an element $x \not\in A$.  We want to define the marginal value $g_A(x)$.  We consider a two-step  random process that first selects a probability $p$ from an appropriate continuous distribution, then a set $B \subseteq A$ by choosing each element of $A$ independently with some probability $p$.  We then define $g$ so that $g_A(x)$ is the expected value of $f_B(x)$ over the random choice of $B$.

Formally, let $P$ be a continuous distribution supported on $[0,1]$ with density given by $ce^{cx}/(e^c - 1)$.  Then, for each $A \subseteq \UU$, we consider the probability distribution $\mu_A$ on $2^A$  given by
\[
\mu_A(B) = \EE_{p \sim P}p^{|B|}(1 - p)^{|A| - |B|}.
\]
Note that this is simply the expectation over our initial choice of $p$ of the probability that the set $B$ is obtained from $A$ by randomly selecting each element of $A$ independently with probability $p$.  Furthermore, for any $A$ and any $A' \subseteq A$, if $B \sim \mu_A$ then $B \cap A' \sim \mu_{A'}$.

Given the distributions $\mu_A$, we shall construct a function $g$ so that
\begin{equation} \label{eq:g-marginals}
g_{A}(x) = \EE_{B \sim \mu_A}[f_B(x)].
\end{equation}
That is, the marginal value $g_{A}(x)$ is the expected marginal gain in $f$ obtained when $x$ is added to a random subset of $A$, obtained by the two-step experiment we have just described.

We can obtain some further intuition by considering how the distribution $P$ affects the values defined in \eqref{eq:g-marginals}.  In the extreme example in which $p = 1$ with probability 1, we have $g_{A}(x) = f_A(x)$ and so $g$ behaves exactly like the original submodular function.  Similarly, if $p = 0$ with probability 1, then $g_{A}(x) = f_\emptyset(x) = f(\{x\})$ for all $A$, and so $g$ is in fact a linear function.  Thus, we can intuitively think of the distribution $P$ as blending together the original function $f$ with some other ``more linear'' approximations of $f$, which have systematically reduced curvature.  We shall see that our choice of distribution results in a function $g$ that gives the desired locality gap.

It remains to show that it is possible to construct a function $g$ whose marginals satisfy \eqref{eq:g-marginals}.  In order to do this, we first note that the probability $\mu_A(B)$ depends only on $|A|$ and $|B|$.  Thus, if we define the values
\begin{equation*}
  \label{eq:m-coeffs}
\mcoeff{a}{b} = \EE_{p \sim P} p^b(1-p)^{a - b} = \int_{0}^1 \frac{ce^{cp}}{e^c - 1}\cdot p^b(1-p)^{a - b}\, dp
\end{equation*}
for all $a, b \ge 0$, then we have $\mu_A(B) =\mcoeff{|A|}{|B|}$.  We adopt the convention that $m_{a,b} = 0$ if either $a$ or $b$ is negative.  Then, we consider the function $g$ given by:
\begin{equation}
 g(A) = \sum_{B \subseteq A} \mcoeff{|A| - 1}{|B| - 1}f(B).
\label{eq:g-direct-definition}
\end{equation}
The marginals of this function are given by
\begin{align*}
g_A(x) &= g(A + x) - g(A) \\
&=\sum_{B \subseteq A+x} \mcoeff{|A|}{|B| - 1}f(B) - \sum_{B \subseteq A} \mcoeff{|A|-1}{|B| - 1}f(B) \\
&=\sum_{B \subseteq A} \left(\mcoeff{|A|}{|B|-1} - \mcoeff{|A|-1}{|B|-1}\right)f(B)
    + \mcoeff{|A|}{|B|}f(B + x).
\end{align*}
The term $\mcoeff{a}{b-1} - \mcoeff{a-1}{b-1}$ evaluates to
\begin{align*}
 \mcoeff{a}{b-1} - \mcoeff{a-1}{b-1} &= \EE_{p \sim P} [p^{b-1} (1-p)^{a-b+1} - p^{b-1} (1-p)^{a-b}] \\ &= \EE_{p \sim P}[-p^b (1-p)^{a-b}] \\ &= -\mcoeff{a}{b}.
\end{align*}
We conclude that
\begin{align*}
 g_A(x) &= \sum_{B \subseteq A} -\mcoeff{|A|}{|B|} f(B) + \mcoeff{|A|}{|B|}f(B + x) \\ &= \sum_{B \subseteq A} \mcoeff{|A|}{|B|} f_B(x) \\ &= \EE_{B \sim \mu_A}[f_B(x)].
\end{align*}

The values $\mcoeff{a}{b}$ used to define $g$ in \eqref{eq:g-direct-definition} can be computed from the following recurrence, which will also play a role in our analysis of the locality gap of Algorithm \ref{alg:1}.
\begin{lemma} \label{lem:m-recurrence}
$\mcoeff{0}{0} = 1$, and for $a > 0$ and $0 \leq b \leq a$,
 \[
  c\mcoeff{a}{b} = 
  (a-b)\mcoeff{a-1}{b} - b \mcoeff{a-1}{b-1} +
  \begin{cases}
  -c/(e^c-1) & \text{if } a = 0, \\
  0 & \text{if } 0 < a < b, \\
  ce^c/(e^c-1) & \text{if } a = b.
  \end{cases}
 \]
\end{lemma}
\begin{proof}
For the base case, we have
\[
m_{0,0} = \int_{0}^1\frac{ce^{cp}}{e^c - 1}\,dp = 1.
\]
The proof of the general case follows from a simple integration by parts:
\begin{align*}
 c\mcoeff{a}{b} &= c \int_0^1 \frac{ce^{cp}}{e^c-1} \cdot p^b (1-p)^{a-b} \, dp \\ &=
 \left.c \cdot \frac{e^{cp}}{e^c-1} \cdot p^b (1-p)^{a-b}\right|^{p =1}_{p=0} -
 c \int_0^1 \left[ bp^{b-1} (1-p)^{a-b} - (a-b)p^b (1-p)^{a-b-1} \right]\frac{e^{cp}}{e^c-1} \, dp \\ &=
 \frac{\ib{a = b} ce^c - \ib{b = 0} c}{e^c-1} + (a-b) \mcoeff{a-1}{b} - b \mcoeff{a-1}{b-1}. \qedhere
\end{align*}
\end{proof}

In future proofs, we shall also need the following upper bound on the sum of the coefficients appearing in \eqref{eq:g-direct-definition}.  Define
\[
\tau(A) = \sum_{B \subseteq A}\mcoeff{|A| - 1}{|B| - 1}.
\]
\begin{lemma}\label{lem:mcoeff-bound}
For all $A \subseteq \UU$,
\[
\tau(A) \le \frac{ce^c}{e^c-1}H_{|A|}
\]
\end{lemma}
\begin{proof}
Expanding the definition of $\mcoeff{|A|-1}{|B|-1}$ we obtain
 \begin{align*}
  \sum_{B \subseteq A} \mcoeff{|A| - 1}{|B| - 1} &=
  \sum_{k=1}^{|A|} \binom{|A|}{k} \mcoeff{|A|-1}{k-1} \\ &=
  \sum_{k=1}^{|A|} \binom{|A|}{k} \int_0^1  \frac{ce^{cp}}{e^c-1} \cdot p^{k-1} (1-p)^{|A|-k} \, dp  \\ & \le
  \frac{ce^c}{e^c-1} \sum_{k=1}^{|A|} \binom{|A|}{k} \int_0^1 p^{k-1} (1-p)^{|A|-k} \, dp  \\ &=
  \frac{ce^c}{e^c-1} \sum_{k=1}^{|A|} \binom{|A|}{k}\frac{(k-1)!(|A| - k)!}{(|A|)!}  \\ &=
  \frac{ce^c}{e^c-1} \sum_{k=1}^{|A|} \frac{1}{k}, 
 \end{align*}
where in the penultimate line, we have used Euler's Beta integral:
\[
\Beta(x,y) = \int_{0}^1 t^{x-1}(1-t)^{y-1}\,dt = \frac{(x-1)!(y-1)!}{(x + y - 1)!},
\] 
whenever $x,y$ are positive integers.
\end{proof}

\subsection{Properties of \texorpdfstring{$g$}{g}} \label{sec:g-properties}

We now show that our potential function $g$ shares many basic properties with $f$.

\begin{lemma} \label{lem:g-properties}
 The function $g$ is normalized, monotone, submodular and has curvature at most $c$.
\end{lemma}
\begin{proof}
From \eqref{eq:g-direct-definition} we have $g(\emptyset) = \mcoeff{-1}{-1}f(\emptyset) = 0$.  Thus, $g$ is normalized.    Additionally, \eqref{eq:g-marginals} immediately implies that $g$ is monotone, since the monotonicity of $f$ implies that each term $f_B(x)$ is non-negative.  Next, suppose that $A_1 \subseteq A_2$ and $x \notin A_2$. Then from \eqref{eq:g-marginals}, we have
 \[
  g_{A_2}(x) = \EE_{B \sim \mu_{A_2}} f_B(x) \leq
  \EE_{B \sim \mu_{A_2}} f_{B \cap A_1}(x) =
  \EE_{B \sim \mu_{A_1}} f_B(x) = g_{A_1}(x),
 \]
where the inequality follows from submodularity of $f$.  Thus, $g$ is submodular. Finally, for any set $A \subseteq \UU$ and any element $x \notin A$, we have
 \[
  g_A(x) = \EE_{B \sim \mu_A} f_B(x) \geq (1-c) f(x) = (1-c) g(x),
 \]
where the inequality follows from the bound on the curvature of $f$, and the second equation from setting $A = \emptyset$ in \eqref{eq:g-marginals}.  Thus, $g$ has curvature at most $c$.
In fact, it is possible to show that for any given $|A|$, $g$ has slightly lower curvature than $f$, corresponding to our intuition that the distribution $P$ blends together $f$ and various functions of reduced curvature.  For our purposes, however, an upper bound of $c$ is sufficient.
\end{proof}

Finally, we note that for any $S \subseteq \UU$, it is possible to bound the value $g(S)$ relative to $f(S)$.

\begin{lemma} \label{lem:g-magnitude}
 For any $A \subseteq \UU$,
 \[ f(A) \leq g(A) \leq \frac{ce^c}{e^c-1} H_{|A|} f(A). \]
\end{lemma}
\begin{proof}
Let $A = \{a_1,\ldots,a_{|A|}\}$ and define $A_i = \{a_1,\ldots,a_i\}$ for $0 \le i \le |A|$.  The formula \eqref{eq:g-marginals} implies that 
 \[ g_{A_i}(a_{i+1}) = \EE_{B \sim \mu_{A_i}} f_B(a_{i+1}) \geq f_{A_i}(a_{i+1}). \]
Summing the resulting inequalities for $i = 0$ to $|A| - 1$, we get
\[
g(A) - g(\emptyset) \ge f(A) - f(\emptyset).
\]
The lower bound then follows from the fact that both $g$ and $f$ are normalized, so $g(\emptyset) = f(\emptyset) = 0$.

For the upper bound, \eqref{eq:g-direct-definition} and monotonicity of $f$ imply that
 \[
  g(A) = \sum_{B \subseteq A}\mcoeff{|A|-1}{|B|-1}f(B) \leq f(A) \sum_{B \subseteq A} \mcoeff{|A|-1}{|B|-1}.
 \]
The upper bound then follows directly from applying the bound of Lemma \ref{lem:mcoeff-bound} to the final sum.
\end{proof}




\subsection{Approximating \texorpdfstring{$g$}{g} via Sampling} \label{sec:g-evaluation}

Evaluating $g(A)$ exactly requires evaluating $f$ on all subsets $B \subseteq A$, and so we cannot compute $g$ directly without using an exponential number of calls to the value oracle $f$.  We now show that we can efficiently estimate $g(A)$ by using a sampling procedure that requires evaluating $f$ on only a polynomial number of sets $B \subseteq A$.  In Section \ref{sec:rp}, we show how to use this sampling procedure to obtain a randomized variant of Algorithm \ref{alg:1} that runs in polynomial time.

We have already shown how to construct the function $g$, and how to interpret the marginals of $g$ as the expected value of a certain random experiment.  Now we show that the direct definition of $g(A)$ in \eqref{eq:g-direct-definition} can also be viewed a the result of a random experiment.  

For a set $A$, consider the distribution $\nu_A$ on $2^{A}$ given by
\[
\nu_A(B) = \frac{\mcoeff{|A|-1}{|B|-1}}{\tau(A)}.
\]
Then, recalling the direct definition of $g$, we have:
\begin{equation*}
 g(A) = \sum_{B \subseteq A} \mcoeff{|A| - 1}{|B| - 1}f(B) 
        = \tau(A)\EE_{B \sim \nu_A}[f(B)]
\end{equation*}

We can estimate $g(A)$ to any desired accuracy by sampling from the distribution $\nu_A$.  Let $B_1,\ldots,B_N$ be $N$ independent random samples from $\nu_A$.  Then, we define:
\begin{equation}
\gapp(A) = \tau(A)\frac{1}{N}\sum_{i = 1}^Nf(B_i)
\end{equation}

\begin{lemma} \label{lem:g-sampling}
Choose $M,\epsilon > 0$, and set
\[N = \frac{1}{2}\left(\frac{ce^c}{e^c - 1}\cdot \frac{H_n}{\epsilon}\right)^2\ln M.\]
Then,
\[ \Pr[|\gapp(A) - g(S)| \geq \epsilon g(S)] = O\left(M^{-1}\right). \]
\end{lemma}
\begin{proof}
We use the following version of Hoeffding's bound.
\begin{fact}[Hoeffding's bound]
  Let $X_1,\ldots,X_N$ be i.i.d. non-negative random variables bounded by $B$, and let $\avgX$ be their
 average.  Suppose that $\EE\avgX \geq \rho B$. Then, for any $\epsilon > 0$,
\begin{equation*} \label{eq:hoeffding}
\Pr[|\avgX - \EE\avgX| \ge \epsilon \EE\avgX] \leq 2\exp\left( -2\epsilon^2\rho^2N \right).
\end{equation*}
\end{fact}
Consider the random variables $X_i = \tau(A)f(B_i)$.  Because $f$ is monotone and each $B_i$ is a subset of $A$, each $X_i$ is bounded by $\tau(A)f(A)$.  The average $\avgX$ of the values $X_i$ satisfies \[\EE\avgX = g(A) \ge f(A), \]
where the inequality follows from Lemma \ref{lem:g-magnitude}.  Thus, Hoeffding's bound implies that
\begin{equation*}
\Pr[|\avgX - \EE\avgX| \ge \epsilon \EE\avgX] \leq 2\exp \left( -\frac{2\epsilon^2N}{\tau(A)^2} \right).
\end{equation*}
By Lemma \ref{lem:mcoeff-bound} we have $\tau(A) \le \frac{ce^c}{e^c - 1}H_{|A|} \le \frac{ce^c}{e^c - 1}H_n$ and so 
\[
2\exp \left( -\frac{2\epsilon^2N}{\tau(A)^2}\right) \le 
2\exp\left(-\ln M\right) = O\left(M^{-1}\right). \qedhere
\]
\end{proof}

\section{Analysis of Algorithm \ref{alg:1}}
\label{sec:algorithm-1}

We now give a complete analysis of the runtime and approximation performance of Algorithm \ref{alg:1}.   The algorithm has two phases: a greedy phase and a local search phase.  Both phases are guided by the  auxiliary potential function $g$ defined in Section \ref{sec:g}.  As noted in Section \ref{sec:g-evaluation}, we cannot, in general, evaluate $g$ in polynomial time.  We postpone concerns dealing with approximating $g$ by sampling until the next section, and in this section suppose that we are given a value oracle returning $g(A)$ for any set $A \subseteq \UU$.  We then show that Algorithm~\ref{alg:1} requires only a polynomial number of calls to the oracle for $g$.  In this way, we can present the main ideas of the proofs without a discussion of the additional parameters and proofs necessary for approximating $g$ by sampling.  In the next section we use the results of Lemma \ref{lem:g-sampling} to implement an approximate oracle for $g$ in polynomial time, and adapt the proofs given here to obtain a 
randomized, polynomial time algorithm.

Consider an arbitrary input to the algorithm.
Let $S = \{s_1,\ldots,s_r\}$ be the solution returned by Algorithm \ref{alg:1} on this instance and $O$ be an optimal solution to this instance.  It follows directly from the definition of the standard greedy algorithm and the type of exchanges considered by Algorithm \ref{alg:1} that $S$ is a base.  Moreover,  because $f$ is montone, we may assume without loss of generality that $O$ is a base, as well.  We index the elements $o_i$ of $O$ by using the following lemma of Brualdi~\cite{Brualdi-1969}.
\begin{fact}[Brualdi's lemma]
Suppose $A,B$ are two bases in a matroid. There is a bijection $\pi\colon A \rightarrow B$ such that for all $a \in A$, $A - a + \pi(a)$ is a base. Furthermore, $\pi$ is the identity on $A \cap B$.
\end{fact}
The main difficulty in bounding the locality ratio of Algorithm~\ref{alg:1} is that we must bound the ratio $f(S)/f(O)$ stated in terms of $f$, by using only the fact that $S$ is locally optimal with respect to $g$.  Thus, we must somehow relate the values of $f(S)$ and $g(S)$.    In the following theorem relates the values of $f$ and $g$ on arbitrary bases of a matroid.  Later, we shall apply this theorem to $S$ and $O$ to obtain an approximation guarantee both for Algorithm \ref{alg:1} and for the randomized variant presented in the next section.

\begin{theorem}
Let $A = \{a_1,\ldots,a_r\}$ and $B = \{b_1,\ldots,b_r\}$  be any two bases of $\mm$, and suppose that we index the elements of $B$ so that $b_i = \pi(a_i)$, where $\pi : A \to B$ is the bijection guaranteed by Brualdi's lemma.  Then,
\label{thm:fg-related}
 \[
  \frac{ce^c}{e^c-1} f(A) \geq f(B) + \sum_{i=1}^r [g(A) - g(A - a_i + b_i)].
 \]
\end{theorem}
\begin{proof}
The proof of Theorem \ref{thm:fg-related} involves a chain of two inequalities and one equation, each of which we shall prove as a separate lemma.  We consider the quantity:
\begin{equation*}\label{eq:locality-quantity}
\sum_{i=1}^r g_{A - a_i}(a_i).
\end{equation*}
First, we shall show in Lemma \ref{lem:locality-ratio:1} that
\begin{equation*}
\sum_{i=1}^r g_{A - a_i}(a_i) \geq
  \sum_{i=1}^r [g(A) - g(A - a_i + b_i)] + \EE_{T \sim \mu_A} \sum_{i=1}^r f_{T - b_i}(b_i),
 \end{equation*}
and then in Lemma \ref{lem:locality-ratio:2} that
\begin{equation*}
\sum_{i=1}^r f_{T - b_i}(b_i) \geq f(B) - cf(T),
\end{equation*}
for any $T \subseteq A$.  Combining these inequalities, we obtain
\begin{equation}
  \label{eq:locality-eq-1}
  \sum_{i = 1}  ^rg_{A - a_i}(a_i) \ge   \sum_{i=1}^r [g(A) - g(A - a_i + b_i)]  + f(B) - c \EE_{T \sim \mu_A} f(T).
\end{equation}
Next, we show in Lemma \ref{lem:locality-ratio:3} that
\begin{equation}
\label{eq:locality-eq-2}
\sum_{i = 1} ^rg_{A - a_i}(a_i) +  c \EE_{T \sim \mu_A} f(T) = \frac{ce^c}{e^c - 1}f(A).
\end{equation}
Combining \eqref{eq:locality-eq-1} and \eqref{eq:locality-eq-2} completes the proof.
\end{proof}
We now prove each of the necessary claims.
\begin{lemma} \label{lem:locality-ratio:1}
 For all $i \in [r]$,
\[ g_{A - a_i}(a_i) \geq g(A) - g(A - a_i + b_i) + \EE_{T \sim \mu_A} \!\! f_{T - b_i}(b_i). \]
\end{lemma}
\begin{proof}
The proof relies on the characterization of the marginals of $g$ given in \eqref{eq:g-marginals}.  We consider two cases: $b_i \notin A$ and $b_i \in A$.  If $b_i \notin A$ then the submodularity of $g$ implies
 \begin{align*}
  g_{A - a_i}(a_i) &\geq g_{A - a_i + b_i}(a_i) \\ &=
  g(A + b_i) - g(A - a_i + b_i) \\ &=
  g_A(b_i) + g(A) - g(A - a_i + b_i) \\ &=
  g(A) - g(A - a_i + b_i) + \EE_{T \sim \mu_A} \!\! f_T(b_i).
 \end{align*}

On the other hand, when $b_i \in A$, we must have $b_i = \pi(a_i) = a_i$ by the definition of $\pi$.  Then,
 \begin{align*}
g_{A - a_i}(a_i) &= \! \EE_{T \sim \mu_{A - a_i}} \!\!\!\!\!\! f_T(a_i) \\
&= \! \EE_{T \sim \mu_A} \!\! f_{T - a_i}(a_i) \\
&= \! \EE_{T \sim \mu_A} \!\!  f_{T - b_i}(b_i) \\
&= g(A) - g(A) + \EE_{T \sim \mu_A} \!\! f_{T - b_i}(b_i) \\
& = g(A) - g(A - a_i + b_i) + \EE_{T \sim \mu_A} \!\! f_{T - b_i}(b_i), 
\end{align*}
where the second equality follows from the fact that if $T \sim \mu_A$ then $T \cap (A \setminus a_i) \sim \mu_{A - a_i}$.
\end{proof}

\begin{lemma} \label{lem:locality-ratio:2}
 For any $T \subseteq A$,
\[ \sum_{i=1}^r f_{T - b_i}(b_i) \geq f(B) - cf(T). \]
\end{lemma}
\begin{proof}
Our proof relies only on the submodularity and curvature of $f$.
 Let $X = T \cap B$, $T' = T \setminus X$ and $B' = B \setminus X$. Furthermore, let $I(X) \subseteq [r]$ be the set of indices $i$ such that $b_i \in X$. We separate the sum on the left-hand side into two parts, based on whether or not $i \in I(x)$. 

The first part of the sum is
 \[
  \sum_{i \notin I(X)} f_{T-b_i}(b_i) =
  \sum_{i \notin I(X)} f_T(b_i) \geq f_T(B'),
 \]
where the final inequality follows from submodularity of $f$.  Next, using $T \cup B' = B \cup T = B \cup T'$, we get
 \begin{align*}
  f_T(B') &= f(T \cup B') - f(T) \\ &= f(B \cup T') - f(T) \\ &\geq f(B) + (1-c) f(T') - f(T),
 \end{align*}
where the final line follows from the fact that $f$ has curvature at most $c$ and $B \cap T' = \emptyset$. 

The second part of the sum is
 \[
  \sum_{i \in I(X)} f_{T-b_i}(b_i) \geq \sum_{i \in I(X)} (1-c) f(b_i) \geq (1-c) f(X),
 \]
where the first inequality follows from the fact that $f$ has curvature at most $c$ and the second inequality from submodularity of $f$.

 Putting both parts together, we deduce
 \begin{align*}
  \sum_{i=1}^r f_{T - b_i}(b_i) &\geq f(B) + (1-c) f(T') - f(T) + (1-c)f(X) \\ 
&\geq  f(B) + (1-c) f(T) - f(T) \\ &= f(B) - cf(T),
 \end{align*}
where in the second inequality we have used $f(T' \cup X) = f(T)$ and $f(T' \cap X) = f(\emptyset) = 0$ together with submodularity of $f$.
\end{proof}

\begin{lemma} \label{lem:locality-ratio:3}
\begin{equation} \label{eq:locality-ratio:3}
 \sum_{i=1}^r g_{A - a_i}(a_i) + c\EE_{T \sim \mu_A} f(T) =
 \frac{ce^c}{e^c-1} f(A).
\end{equation}
\end{lemma}
\begin{proof}
The proof relies primarily on the recurrence given in Lemma \ref{lem:m-recurrence} for the values $\mcoeff{a}{b}$ used to define $g$.  From the characterization of the marginals of $g$ given in \eqref{eq:g-marginals} we have
 \[
  g_{A - a_i}(a_i) = \EE_{T \sim \mu_{A - a_i}} [f_T(a_i)] = \EE_{T \sim \mu_{A - a_i}} [f(T + a_i) - f(T)].
 \]
Each subset $D \subseteq A$ appears in the expectation.  Specifically, if $a_i \in D$ then we have the term $\mu_{A - a_i}(D - a_i)f(D)$, and if $a_i \in A \setminus D$ then we have the term $-\mu_{A - a_i}(D)f(D)$.  Therefore the coefficient of $f(D)$ in the left-hand side of~\eqref{eq:locality-ratio:3} is thus given by
\[
\left(\sum_{a_i \in D} \mu_{A - a_i}(D - a_i)\right) - \Biggl(\sum_{a_i \notin D} \mu_{A - a_i}(D)\Biggr) + c\mu_A(D)  =
|D|\mcoeff{r-1}{|D| - 1} - (r-|D|)\mcoeff{r-1}{|D|} + c\mcoeff{r}{|D|}.
 \]
According to the recurrence for $\mcoeff{}{}$ given in Lemma~\ref{lem:m-recurrence}, the right-hand side vanishes unless $D = \emptyset$, in which case it is $\frac{-c}{e^c - 1}f(\emptyset) = 0$, or $D = A$, in which case it is $\frac{ce^c}{e^c-1}f(A)$.
\end{proof}

We are now ready to prove this section's main claim, which gives bounds on both the approximation ratio and complexity of Algorithm 1.

\begin{theorem}
\label{thm:alg-1}
Algorithm \ref{alg:1} is a $\left(\frac{1 - e^{-c}}{c} - \epsilon\right)$-approximation algorithm,
requiring at most $O(r^2n\epsilon^{-1}\log n)$ evaluations of $g$.
\end{theorem}
\begin{proof}
We first consider the number of evaluations of $g$ required by Algorithm \ref{alg:1}.
The initial greedy phase requires $O(rn)$ evaluations of $g$, as does each iteration of the local search phase.  Thus, the total number of evaluations of $g$ required by Algorithm \ref{alg:1} is $O(rnI)$, where $I$ is the number of improvements applied in the local search phase.  We now derive an upper bound on $I$.

Let $g^* = \max_{A \in \II}g(A)$ be the maximum value attained by $g$ on any independent set in $\mm$.
Algorithm \ref{alg:1} begins by setting $S$ to a greedy solution $\sinit$, and each time it selects an improved solution $S'$ to replace $S$ by, we must have
\[
g(S') > (1 + \epsilon_1)g(S)
\]
Thus, the number of improvements that Algorithm \ref{alg:1} can apply is at most
\[
\log_{1 + \epsilon_1}\frac{g^*}{g(\sinit)}.
\]

Fisher, Nemhauser, and Wolsey \cite{Fisher-1978} show that the greedy algorithm is a $1/2$-approximation algorithm for maximizing any monotone submodular function subject to a matroid constraint.  In particular, because $g$ is monotone submodular, as shown in Lemma \ref{lem:g-properties}, we must have
\[
I \le \log_{1 + \epsilon_1}\frac{g^*}{g(\sinit)} \le \log_{1 + \epsilon_1}2 = O(\epsilon_1^{-1}) = O(rH_r\epsilon^{-1}) = O(r\epsilon^{-1}\log n).
\]

Next, we consider the approximation ratio of Algorithm \ref{alg:1}.  
Recall that $O$ is an optimal solution of the arbitrary instance $(\mm=(\UU,\II),f)$ on which Algorithm \ref{alg:1} returns the solution $S$.  We apply Theorem \ref{thm:fg-related} to the bases $S$ and $O$, indexing $S$ and $O$ as in the theorem so that $S - s_i + o_i \in \II$ for all $i \in [r]$, to obtain:
\begin{equation}
\label{eq:fg-related-1}
  \frac{ce^c}{e^c-1} f(S) \geq f(O) + \sum_{i=1}^r [g(S) - g(S - s_i + o_i).]
\end{equation}
Then, we note that we must have
\[g(S - s_i + o_i) \le (1 + \epsilon_1)g(S)\]
for each value $i \in [r]$---otherwise, Algorithm \ref{alg:1} would have exchanged $s_i$ for $o_i$ rather than returning $S$.  Summing the resulting $r$ inequalities gives
\[\sum_{i = 1}^r[g(S) - g(S - s_i + o_i)] \ge -r\epsilon_1g(S).\]
Applying this and upper bound on $g(S)$ from Lemma \ref{lem:g-magnitude} to~\eqref{eq:fg-related-1} we then obtain
\[
\frac{ce^c}{e^c-1} f(S) \ge f(O) - r\epsilon_1 g(S) \ge f(O) - \frac{ce^c}{e^c - 1}r\epsilon_1H_rf(S)
\ge f(O) - \frac{ce^c}{e^c - 1}r\epsilon_1H_r f(O).
\]
Rewriting this inequality using the definition $\epsilon_1 = \frac{\epsilon}{rH_r}$ then gives
\[
f(S) \ge \left(\frac{1 - e^{-c}}{c} - \epsilon\right)f(O),
\]
and so Algorithm \ref{alg:1} is a $\left(\frac{1 - e^{-c}}{c} - \epsilon\right)$-approximation algorithm.
\end{proof}

\section{A randomized, polynomial-time algorithm}
\label{sec:rp}
Our analysis of Algorithm \ref{alg:1} supposed that we were given an oracle for computing the value of the potential function $g$.  We now use the results of Lemma \ref{lem:g-sampling}, which shows that the value $g(A)$ can be approximated for any $A$ by using a polynomial number of samples, to implement a randomized, polynomial-time approximation algorithm that does not require an oracle for $g$.  The resulting algorithm attains the same approximation ratio as Algorithm \ref{alg:1} with high probability.

The modified algorithm is shown in Algorithm \ref{alg:2}.  Algorithm \ref{alg:2} uses an approximation $\gapp$ of $g$ that is obtained by taking $N$ independent random samples of $f$ each time $\gapp$ is calculated.  The number of samples $N$ depends on the parameters $\epsilon$ and $\alpha$, in addition to the rank $r$ of $\mm$ the size $n$ of $\UU$.  As in Algorithm \ref{alg:1}, $\epsilon$ governs how much an exchange must improve the current solution before it is applied, and so affects both the approximation performance and runtime of the algorithm.  The additional parameter $\alpha$ controls the probability that Algorithm \ref{alg:2} fails to produce a $\left(\frac{1 - e^{-c}}{c} - \epsilon\right)$-approximate solution.  Specifically, we show that Algorithm \ref{alg:2} fails with probability at most $O(n^{-\alpha})$.  

For the analysis, we assume that $\epsilon \leq 1$ and $r \geq 2$, which imply that $\epsilon_2 \leq 1/12$.

\begin{algorithm}
\label{alg:2}
\KwIn{$\mm = (\UU,\II),\ f,\ c,\ \epsilon,\ \alpha$}
Set $\epsilon_2 = \frac{\epsilon}{4rH_r}$\;
Set $I = \left(\left(\frac{1 + \epsilon_2}{1 - \epsilon_2}\right)(2 + 3r\epsilon_2) - 1\right)\epsilon_2^{-1}$\;
Set $N = \frac{1}{2}\left(\frac{ce^c}{e^c - 1} \cdot \frac{H_n}{\epsilon_2}\right)^2 \ln\left((I + 1)rn^{1 + \alpha}\right)$\;
Let $\gapp$ be an approximation to $g$ computed by taking $N$ random samples\;
Let $\sinit$ be the result of running the standard greedy algorithm on $(\mm, \gapp)$\;
$S \gets \sinit$\;
$v \gets \gapp(\sinit)$\;
\For{$i \gets 1$ \KwTo $I$}{
  $\mathrm{done} \gets \mathbf{true}$\;
  \ForEach{element $e \in S$ and $x \in \UU \setminus S$}{
    $S' \gets S - e + x$\;
    \If{$S' \in \II$}{
      $v' \gets \gapp(S')$\;
      \If{$v' > (1 + \epsilon_2)v$}{
        $v \gets v'$ and $S \gets S'$\;
        $\mathrm{done} \gets \mathbf{false}$\;
        \Break;
      }
    }
  }
  \lIf{$\mathrm{done}$}{\Return $S$}
}
\Return \texttt{Error}\;
\caption{The non-oblivious local search algorithm}
\end{algorithm}

The local search routine in Algorithm \ref{alg:2} runs some number $I$ of iterations, signaling an error if it fails to converge to a local optimum after this many improvements.  In each iteration, the algorithm searches through all possible solutions $S' = S - e + x$, sampling the value $\gapp(S')$ if $S' \in \II$.  If the sampled value of $\gapp(S')$ exceeds the sampled value for $\gapp(S)$ by a factor of at least $(1 + \epsilon_2)$, the algorithm updates $S$ and moves to the next iteration.  Otherwise, it returns the current solution.  Note that we store the last sampled value $\gapp(S)$ of the current solution in $v$, rather than resampling $\gapp(S)$ each time we check an improvement $S'$.

The analysis of Algorithm \ref{alg:2} follows the same general pattern as that presented in the previous section.  Here however, we must address the fact that $\gapp$ does not always agree with $g$.  First, we estimate the probability that all of the computations of $\gapp$ made by Algorithm \ref{alg:2} are reasonably close to the value of $g$.

\begin{lemma}
\label{lem:sampling-assumption}
With probability $1 - O(n^{-\alpha})$, we have $|\gapp(A) - g(A)| \le \epsilon_2g(A)$
for all sets $A$ for which Algorithm \ref{alg:2} computes $\gapp(A)$.
\end{lemma}
\begin{proof}
We first bound the total number of sets $A$ for which Algorithm \ref{alg:2} computes $\gapp(A)$.  The initial greedy phase requires fewer than $rn$ evaluations, as does each of the $I$ iterations of the local phase.  The total number of evaluations is therefore less than $(I + 1)rn$.

Algorithm \ref{alg:2} uses 
\[N = \frac{1}{2}\left(\frac{ce^c}{e^c - 1} \cdot \frac{H_n}{\epsilon_2}\right)^2 \ln\left((I + 1)rn^{1 + \alpha}\right)
\]
samples for every computation of $\gapp(A)$.  By Lemma \ref{lem:g-sampling}, the probability that we have $|g(A) - \gapp(A)| \ge \epsilon_2 g(A)$ for any given set $A$ is then $O\left(\frac{1}{(I + 1)rn^{1 + \alpha}}\right)$.  From the union bound, then, the probability that at least one of the $(I + 1)rn$ sets $A$ for which Algorithm \ref{alg:2} computes $\gapp(A)$ does not satisfy the desired error bound is at most $O\left(\frac{(I + 1)rn}{(I + 1)rn^{1 + \alpha}}\right) = O(n^{-\alpha})$.
\end{proof}

We call the condition that $|g(A) - \gapp(A)| \le \epsilon_2 g(A)$ for all sets $A$ considered by Algorithm \ref{alg:2} the \emph{sampling assumption}.  Lemma \ref{lem:sampling-assumption} shows that the sampling assumption holds with high probability.  

Now, we must adapt the analysis of Section \ref{sec:algorithm-1}, which holds when $g$ is computed exactly, to the setting in which $g$ is computed approximately.  In Theorem \ref{thm:alg-1}, we showed that $g(\sinit)$ is within a constant factor of the largest possible value that $g$ could take on any set $A \subseteq \UU$.  Then, because the algorithm always improved $g$ by a factor of at least $(1 + \epsilon_1)$, we could bound the number of local search iterations that it performed.  Finally, we applied Theorem \ref{thm:fg-related} to translate the local optimality of $S$ with respect to $g$ into a lower bound on $f(S)$.

Here we follow the same general approach.  First, we derive the following result, which shows that the initial value $\gapp(\sinit)$ is within a constant factor of the maximum value $\gapp^*$ of $\gapp(A)$ on any set $A$ considered by Algorithm \ref{alg:2}.\footnote{A similar result for the greedy algorithm applied to an approximately calculated submodular function is given by Calinescu et al.~\cite{Calinescu-2011}.  However, in their model, the \emph{marginals} of a submodular function are approximately calculated, while in ours, the \emph{value} of the submodular function is approximately calculated.  For the sake of completeness, we provide a complete proof for our setting.}
\begin{lemma} \label{lem:greedy-phase-2}
Suppose that the sampling assumption is true, and let $\gapp^*$ be the maximum value of $\gapp(A)$ over all sets $A$ considered by Algorithm \ref{alg:2}.  Then, 
 \[
(2 + 3r\epsilon_2)\left(\frac{1 + \epsilon_2}{1 - \epsilon_2}\right)\gapp(\sinit) \geq \gapp^*.
 \]
\end{lemma}
\begin{proof}
The standard greedy algorithm successively chooses a sequence of sets $\emptyset = S_0,S_1,\ldots,S_r = \sinit$, where each $S_i$ for $i > 0$ satisfies $S_i = S_{i - 1} + s_{i}$ for some element $s_i \in \UU \setminus S_{i-1}$.  The element $s_i$ is chosen at each phase according to the formula
\[ s_i = \argmax_{\substack{x\,\in\,\UU \setminus S_{i-1} \\ \text{s.t.}\,S_{i - 1} + x\,\in\,\II}} \gapp(S_{i-1} + x). \]

Let $O$ be any base of $\mm$ on which $g$ attains its maximum value.  According to Brualdi's lemma, we can index $O = \{o_1,\ldots,o_r\}$ so that $o_i = \pi(s_i)$ for all $i \in [r]$.  Then, the set $S_{i - 1} + o_i$ is independent for all $i \in [r]$.  Thus, we must have
\[
\gapp(S_{i-1} + s_i) \ge \gapp(S_{i-1} + o_i)
\]
for all $i \in [r]$.  In order to use monotonicity and submodularity, we translate this into an inequality for $g$.  From the sampling assumption, we have
\[
(1 + \epsilon_2)g(S_{i-1} + s_i) \ge \gapp(S_{i-1} + s_i) \ge \gapp(S_{i-1} + o_i) \ge (1-\epsilon_2)g(S_{i-1} + o_i).
\]
Then, since $(1+\epsilon_2)/(1-\epsilon_2) \leq 1+3\epsilon_2$ for all $\epsilon_2 \leq 1/3$,
\[
(1 + 3\epsilon_2)g(S_{i-1} + s_i) \ge \frac{(1+\epsilon_2)}{(1-\epsilon_2)}g(S_{i-1} + s_i)  \ge g(S_{i-1} + o_i).
\]
Subtracting $g(S_{i-1})$ from each side above, we obtain
\[3\epsilon_2 g(S_i) + g_{S_{i-1}}(s_i) \ge g_{S_{i-1}}(o_i) \]
for each $i \in [r]$.  Summing the resulting $r$ inequalities, we obtain a telescoping summation, which gives
\begin{equation*}
\label{eq:greedy-1}
3\epsilon_2\sum_{i = 1}^rg(S_i) + g(\sinit) \ge \sum_{i=1}^rg_{S_{i-1}}(o_i) \ge \sum_{i = 1}^rg_{\sinit}(o_i) \ge g_{\sinit}(O) = g(O \cup \sinit) - g(\sinit),
\end{equation*}
where we have used submodularity of $g$ for the second and third inequalities.  Then, using monotonicity of $g$, we have $3\epsilon_2 \sum_{i = 1}^rg(\sinit) \ge 3\epsilon_2 \sum_{i = 1}^rg(S_i)$ on the left, and
$g(O \cup \sinit) \ge g(\sinit)$ on the right, and so
\begin{equation}
\label{eq:g-greedy-approx}
3r\epsilon_2 g(\sinit) + 2g(\sinit) \ge g(O).
\end{equation}

Finally, by the sampling assumption we must have $\gapp(\sinit) \ge (1 - \epsilon_2)g(\sinit)$ and also $(1 + \epsilon_2)g(O) \ge (1 + \epsilon_2)g(A) \ge \gapp(A)$ for any set $A$ considered by the algorithm.  Thus, \eqref{eq:g-greedy-approx} implies
\[
(2 + 3r\epsilon_2)\left(\frac{1 + \epsilon_2}{1 - \epsilon_2}\right)\gapp(\sinit) \ge \gapp^*. \qedhere
\]
\end{proof}

The next difficulty we must overcome is that the final set $S$ produced by Algorithm \ref{alg:2} is (approximately) locally optimal only with respect to the sampled function $\gapp(S)$.  In order to use Theorem  \ref{thm:fg-related} to obtain a lower bound on $f(S)$, we must show that $S$ is approximately locally optimal with respect to $g$ as well.  We accomplish this in our next lemma, by showing that any significant improvement in $\gapp$ must correspond to a (somewhat less) significant improvement in $g$.

\begin{lemma} \label{lem:local-optimality-2}
Suppose that the sampling assumption holds and that $\gapp(A) \le (1 + \epsilon_2)\gapp(B)$ for some pair of sets $A,B$ considered by Algorithm \ref{alg:2}.  Then
\[
g(A) \le (1 + 4\epsilon_2)g(B)
\]
\end{lemma}
\begin{proof}
From the sampling assumption, we have
\[
(1 - \epsilon_2)g(A) \le \gapp(A) \le (1 + \epsilon_2)\gapp(B) \le (1 + \epsilon_2)(1 + \epsilon_2)g(B).
\]
Thus
\[
 g(A) \leq \frac{(1+\epsilon_2)^2}{1-\epsilon_2} g(B) \leq (1 + 4\epsilon_2)g(B),
\]
where the second inequality holds since $\epsilon_2 \leq 1/5$.
\end{proof}

We now prove our main result.
\begin{theorem}
\label{thm:alg-2}
Algorithm 2 runs in time $\tilde{O}(r^4n\epsilon^{-3}\alpha)$ and returns a $\left(\frac{1-e^{-c}}{c} - \epsilon\right)$-approximation with probability $1 - O(n^{-\alpha})$.
\end{theorem}
\begin{proof}
As in the proof of Theorem \ref{thm:alg-1}, we consider some arbitrary instance $(\mm=(\UU,\II),f)$ of monotone submodular matroid maximization with upper bound $c$ on the curvature of $f$, and let $O$ be an optimal solution of this instance.  We shall show that if the sampling assumption holds, Algorithm \ref{alg:2} returns a solution $S$ satisfying $f(S) \ge \left(\frac{1 - e^{-c}}{c} - \epsilon\right)f(O)$.  Then, Lemma \ref{lem:sampling-assumption} shows that this happens with probability $1 - O(n^{-\alpha})$.

As in Algorithm \ref{alg:2}, set
\[
I = \left(\left(\frac{1 + \epsilon_2}{1 - \epsilon_2}\right)(2 + 3r\epsilon_2) - 1\right)\epsilon_2^{-1}.
\]
Suppose that the sampling assumption holds, and let $g^*$ be the maximum value taken by $\gapp(A)$ for any set $A$ considered by Algorithm \ref{alg:2}.  At each iteration of Algorithm \ref{alg:2}, either a set $S$ is returned, or the value $v$ is increased by a factor of at least $(1 + \epsilon_2)$.  Suppose that the local search phase of Algorithm \ref{alg:2} fails to converge to a local optimum after $I$ steps, and so does not return a solution $S$.  Then we must have
\[
v \ge (1 + \epsilon_2)^I\gapp(\sinit) > (1 + I\epsilon_2)\gapp(\sinit) =
\left(\frac{1 + \epsilon_2}{1 - \epsilon_2}\right)(2 + 3r\epsilon_2)\gapp(\sinit) \ge g^*,
\]
where the last inequality follows from Lemma \ref{lem:greedy-phase-2}.  But, then we must have $\gapp(A) > g^*$ for some set $A$ considered by the algorithm.  Thus Algorithm \ref{alg:2} must produce a solution $S$.  

As in Theorem \ref{thm:alg-1}, we apply Theorem \ref{thm:fg-related} to the bases $S$ and $O$, indexing $S$ and $O$ as in the theorem so that $S - s_i + o_i \in \II$ for all $i \in [r]$, to obtain:
\begin{equation}
\label{eq:fg-related-2}
  \frac{ce^c}{e^c-1} f(S) \geq f(O) + \sum_{i=1}^r [g(S) - g(S - s_i + o_i).]
\end{equation}
Then, since Algorithm \ref{alg:2} returned $S$, we must then have:
\[
\gapp(S - s_i + o_i) \le (1 + \epsilon_2)\gapp(S)
\]
for all $i \in [r]$.  From Lemma \ref{lem:local-optimality-2} we then have
\[
g(S - s_i + o_i) \le (1 + 4\epsilon_2)g(S)
\]
for all $i \in [r]$.  Summing the resulting $r$ inequalities gives
\[
\sum_{i = 1}^r\left[g(S) - g(S - s_i + o_i) \right]\ge -4r\epsilon_2g(S).
\]
Applying Theorem \ref{thm:alg-1} and the upper bound on $g(S)$ from Lemma \ref{lem:g-magnitude} in \eqref{eq:fg-related-2}, we then have
\[
\frac{ce^c}{e^c - 1}f(S) \ge f(O) - 4r\epsilon_2g(S) \ge 
f(O) - \frac{ce^c}{e^c - 1}4r\epsilon_2H_rf(S) \ge f(O) - \frac{ce^c}{e^c - 1}4r\epsilon_2H_rf(O).
\]
Rewriting this inequality using the definition $\epsilon_2 = \frac{\epsilon}{4rH_r}$ then gives
\[
f(S) \ge \left(\frac{1 - e^{-c}}{c}-\epsilon\right)f(O).
\]

The running time of Algorithm \ref{alg:2} is dominated by the number of calls it makes to the value oracle for $f$. We note, as in the proof of Lemma \ref{lem:sampling-assumption}, that the algorithm evaluates $\gapp(A)$ on $O(rnI)$ sets $A$.  Each evaluation requires $N$ samples of $f$, and so the resulting algorithm requires 
\[
O(rnIN) = \tilde{O}(rn \epsilon_2^{\,-3}\alpha) = \tilde{O}(r^4n\epsilon^{-3}\alpha)
\]
calls to the value oracle for $f$.
\end{proof}

\section{Extensions} \label{sec:extensions}

The algorithm presented in Section \ref{sec:algorithm} produces a $(1-e^{-c})/c - \epsilon$ approximation for any $\epsilon > 0$, and it requires knowledge of $c$. In this section we show how to produce a clean $(1-e^{-c})/c$ approximation, and how to dispense with the knowledge of $c$. Unfortunately, we are unable to combine both improvements for technical reasons.

It will be useful to define the function
\[ \rho(c) = \frac{1-e^{-c}}{c}, \]
which gives the optimal approximation ratio.

\subsection{Clean approximation} \label{sec:clean-approximation}

In this section, we assume $c$ is known, and our goal is to obtain a $\rho(c)$ approximation algorithm. We accomplish this by combining the algorithm from Section \ref{sec:rp} with partial enumeration. 

For $x \in \UU$ we consider the contracted matroid $\mm  / x$ on $\UU - x$ whose independent sets are given by $\II_x = \{ A \subseteq \UU - x : A + x \in \II \}$, and the contracted submodular function $(f / x)$ which is given by $(f / x)(A) = f(A + x)$.  It is easy to show that this function is a monotone submodular function whenever $f$ is, and has curvature at most that of $f$.  Then, for each $x \in \UU$, we apply Algorithm \ref{alg:2} to the instance $(\mm/x, f/x)$ to obtain a solution $S_x$.  We then return $\argmax_{x \in \UU}f(S_x)$.

Fisher, Nemhauser, and Wolsey \cite{Nemhauser-1978} analyze this technique in the case of submodular maximization over a uniform matroid, and Khuller, Moss, and Naor \cite{Khuller-1999} make use of the same technique in the restricted setting of budgeted maximum coverage.  Calinescu et al.~\cite{Calinescu-2007} use a similar technique to eliminate the error term from the approximation ratio of the continuous greedy algorithm for general monotone submodular matroid maximization.  Our proof relies on the following general claim.
\begin{lemma} \label{lem:derived-instance}
 Suppose $A \subseteq O$ and $B \subseteq \UU\setminus A$ satisfy $f(A) \geq (1-\theta) f(O)$ and $f_A(B) \geq (1-\theta_B) f_A(O\setminus A)$. Then
 \[ f(A \cup B) \geq (1 - \theta_A \theta_B) f(O). \]
\end{lemma}
\begin{proof}
We have
\begin{align*}
 f(A \cup B) &= f_A(B) + f(A) \\ &\geq (1 - \theta_B) f_A(O \setminus A) + f(A) \\ &= (1 - \theta_B) f(O) + \theta_B f(A) \\ &\geq (1 - \theta_B) f(O) + \theta_B (1 - \theta_A) f(O) \\ &= (1 - \theta_A \theta_B) f(O). \qedhere
\end{align*}
\end{proof}

Using Lemma \ref{lem:derived-instance} we show that the partial enumeration procedure gives a clean $\rho(c)$-approximation algorithm.

\begin{theorem} \label{thm:clean-approximation}
The partial enumeration algorithm runs in time $\tilde{O}(r^7 n^2 \alpha)$, and with probability $1-O(n^{-\alpha})$, the algorithm has an approximation ratio of $\rho(c)$.
\end{theorem}
\begin{proof}
 Let $O = \{o_1,\ldots,o_r\}$ be an optimal solution to some instance $(\mm,f)$. Since submodularity of $f$ implies
 \[ \sum_{i=1}^r f(o_i) \geq f(O), \]
 there is some $x \in O$ such that $f(x) \geq f(O)/r$. Take $A = \{x\}$ and $B = S_x$ in Lemma~\ref{lem:derived-instance}.  Then, from Theorem \ref{thm:alg-2} we have $f(S_x) \ge (\rho(c) - \epsilon)f(O)$ with probability $1 - O(n^{-\alpha})$ for any $\epsilon$.  We set $\epsilon = (1 - \rho(c))/r$.  Then, substituting $\theta_A = 1-1/r$ and $\theta_B = 1 - \rho(c) + (1 - \rho(c))/r$, we deduce that the resulting approximation ratio in this case is
 \begin{align*}
  1 - \left(1-\frac{1}{r}\right)\left(1 - \rho(c) + \frac{1-\rho(c)}{r}\right) &=
  1 - \left(1-\frac{1}{r}\right)\left(1+\frac{1}{r}\right)(1-\rho(c)) \\ &\geq
  1 - (1-\rho(c)) = \rho(c).
 \end{align*}
The partial enumeration algorithm simply runs Algorithm \ref{alg:2} $n$ times, using $\epsilon = O(r^{-1})$ and so its running time is $\tilde{O}(r^7n^2\alpha)$.
\end{proof}

\subsection{Unknown curvature} \label{sec:unknown-curvature}
In this section, we remove the assumption that $c$ is known, but retain the error parameter $\epsilon$. The key observation is that if a function has curvature $c$ then it also has curvature $c'$ for any $c' \geq c$. This, combined with the continuity of $\rho$, allows us to ``guess'' an approximate value of $c$.

Given $\epsilon$, consider the following algorithm. Define the set $C$ of curvature approximations by
\[ C = \{ k \epsilon : 1 \leq k \leq \lfloor \epsilon^{-1} \rfloor \} \cup \{ 1 \}. \]
For each guess $c' \in C$, we run the main algorithm with that setting of $c'$ and error parameter $\epsilon/2$ to obtain a solution $S_{c'}$. Finally, we output the set $S_{c'}$ maximizing $f(S_{c'})$.

\begin{theorem} \label{thm:unknown-curvature}
Suppose $f$ has curvature $c$.
The unknown curvature algorithm runs in time $\tilde{O}(r^4n\epsilon^{-4}\alpha)$, and with probability  $1-O(n^{-\alpha})$, the algorithm has an approximation ratio of $\rho(c)-\epsilon$.
\end{theorem}
\begin{proof}
 From the definition of $C$ it is clear that there is some $c' \in C$ satisfying $c \leq c' \leq c + \epsilon$. Since $f$ has curvature $c$, the set $S_{c'}$ is a $\rho(c') - \epsilon/2$ approximation. Elementary calculus shows that on $(0,1]$, $\rho' \geq -1/2$, and so we have
\[\rho(c') - \epsilon/2 \ge \rho(c + \epsilon) - \epsilon/2 \ge \rho(c) + \epsilon/2 - \epsilon/2 = \rho(c) - \epsilon. \qedhere\]
\end{proof}

\subsection{Maximum Coverage}
\label{sec:maximum-coverage}
In the special case that $f$ is given explicitly as a \emph{coverage function} function, we can evaluate the potential function $g$ exactly in polynomial time.  A (weighted) coverage function is a particular kind of monotone submodular function that may be given in the following way.  There is a universe $\otherU$ with non-negative weight function $w \colon \otherU \rightarrow \RRnn$.  The weight function is extended to subsets of $\otherU$ linearly, by letting $w(S) = \sum_{s \in S}w(s)$ for all $S \subseteq \otherU$.  Additionally, we are given a family $\{\otherUset{a}\}_{a \in \UU}$ of subsets of $\otherU$, indexed by a set $\UU$.  The function $f$ is then defined over the index set $\UU$, and $f(A)$ is simply the total weight of all elements of $\otherU$ that are covered by those sets whose indices appear in $A$.  That is, $f(A) = w\left(\bigcup_{a \in A} \otherUset{a} \right)$.

We now show how to compute the potential function $g$ exactly in this case.  For a set $A \subseteq \UU$ and an element $x \in \otherU$, we denote by $A[x]$ the collection
$\{a \in A : x \in \otherUset{a}\}$ of indices $a$ such that $x$ is in the set $\otherUset{a}$.  Then, recalling the definition of $g(A)$ given in \eqref{eq:g-direct-definition}, we have
\begin{align*}
g(A) &= \sum_{B \subseteq A}\mcoeff{|A| - 1}{|B| - 1}f(B) \\
&= \sum_{B \subseteq A}\mcoeff{|A| - 1}{|B| - 1}\!\!\!\!\sum_{x \in \bigcup_{b \in B}\otherUset{b}}\!\!\!\!w(x) \\
&= \sum_{x \in \otherU}w(x)\!\!\!\!\sum_{\substack{B \subseteq A\ s.t. \\ A[x] \cap B \neq \emptyset}}\!\!\!\!\mcoeff{|A| - 1}{|B| - 1}.
\end{align*}
Consider the coefficient of $w(x)$ in the above expression for $g(A)$.  We have
\begin{align*}
\sum_{\substack{B \subseteq A\ s.t. \\ A[x] \cap B \neq \emptyset}}\!\!\!\!\mcoeff{|A| - 1}{|B| - 1} 
&= \sum_{B \subseteq A}\mcoeff{|A| - 1}{|B| - 1}\ -\!\! \sum_{B \subseteq A \setminus A[x]}\mcoeff{|A| - 1}{|B| - 1} \\
&= \sum_{i = 0}^{|A|}\binom{|A|}{i}\mcoeff{|A| - 1}{i - 1}\ -\!\! \sum_{i = 0}^{|A\setminus A[x]|}\binom{|A \setminus A[x]|}{i}\mcoeff{|A| - 1}{i - 1} \\
&= \sum_{i = 0}^{|A|}\binom{|A|}{i}\EE_{p \sim P}[p^{i-1}(1 - p)^{|A| - i}]\  -\!\! \sum_{i = 0}^{|A \setminus A[x]|}\binom{|A \setminus A[x]|}{i}\EE_{p \sim P}[p^{i-1}(1 - p)^{|A| - i}] \\
&=\!\EE_{p \sim P}\!\left[\frac{1}{p}\sum_{i = 0}^{|A|}\binom{|A|}{i}p^{i}(1 - p)^{|A| - i} - \frac{(1 - p)^{|A[x]|}}{p}\sum_{i = 1}^{|A \setminus A[x]|}\!\!\binom{|A \setminus A[x]|}{i}p^{i}(1 - p)^{|A \setminus A[x]| - i}\right] \\
&=\!\EE_{p \sim P}\!\left[\frac{1 - (1 - p)^{|A[x]|}}{p}\right].
\end{align*}
Thus, if we define
\[
\ell_k = \EE_{p \sim P}\left[\frac{1 - (1 - p)^k}{p}\right]
\]
we have
\[
g(A) = \sum_{x \in \otherU} \ell_{|A[x]|}w(x),
\]
and so to compute $g$, it is sufficient to maintain for each element $x \in \otherU$ a count of the number of sets $A[x]$ with indices in $A$ that contain $x$.  Using this approach, each change in $g(S)$ resulting from adding an element $x$ to $S$ and removing an element $e$ from $S$ during one step of the local search phase of Algorithm \ref{alg:1} can be computed in time $O(|\otherU|)$.

We further note that the coefficients $\ell_k$ are easily calculated using the following recurrence.  For $k = 0$, 
\[
\ell_0 = \EE_{p \sim P}\left[\frac{1 - (1 - p)^0}{p}\right] = 0,
\]
while for $k > 0$, 
\[
\ell_{k+1} = \EE_{p \sim P}\left[\frac{1 - (1 - p)^{k + 1}}{p}\right] = \EE_{p \sim P}\left[\frac{1 - (1 - p)^k + p(1 - p)^k}{p}\right] = \ell_k + \EE_{p \sim P}(1 - p)^k = \ell_k + \mcoeff{k}{0}.
\]
The coefficients $\ell_k$ obtained in this fashion in fact correspond (up to a constant scaling factor) to those used to define the non-oblivious coverage potential in \cite{Filmus-2012}, showing that our algorithm for monotone submodular matroid maximization is indeed a generalization of the algorithm already obtained in the coverage case.
\bibliographystyle{plain}
\bibliography{Papers}
\end{document}